\documentclass[12pt]{article}

\usepackage[T1]{fontenc}
\usepackage{lmodern}

\usepackage{setspace}
\usepackage[margin=1.25in]{geometry}

\usepackage[dvipsnames]{xcolor}
\usepackage{pdfpages}
\usepackage{float}
\usepackage{multirow}
\usepackage{caption}
\usepackage{subcaption}
\usepackage{epigraph}

\usepackage{mathtools}
\usepackage{amssymb}
\usepackage{amsthm}
\usepackage{amsfonts}
\usepackage{aliascnt}
\usepackage{accents}
\usepackage{dutchcal}
\usepackage{bbm}

\usepackage{enumitem}
\usepackage{sgame}
\usepackage{tikz}
\usepackage{tikz-cd}
\usetikzlibrary{calc,shapes,arrows}
\usepackage{tcolorbox}
\usepackage{listings}
\usepackage[normalem]{ulem}

\usepackage[round]{natbib}
\newcommand{\E}{\mathbb{E}}

\newtheorem{theorem}{Theorem}[section]

\newaliascnt{proposition}{theorem}
\newtheorem{proposition}[proposition]{Proposition}
\aliascntresetthe{proposition}

\newaliascnt{lemma}{theorem}
\newtheorem{lemma}[lemma]{Lemma}
\aliascntresetthe{lemma}

\newaliascnt{corollary}{theorem}
\newtheorem{corollary}[corollary]{Corollary}
\aliascntresetthe{corollary}

\newaliascnt{claim}{theorem}

\aliascntresetthe{claim}

\theoremstyle{definition}

\newaliascnt{definition}{theorem}

\aliascntresetthe{definition}

\newaliascnt{example}{theorem}

\aliascntresetthe{example}

\newaliascnt{assumption}{theorem}
\newtheorem{assumption}[assumption]{Assumption}
\aliascntresetthe{assumption}

\newaliascnt{condition}{theorem}

\aliascntresetthe{condition}

\newaliascnt{question}{theorem}

\aliascntresetthe{question}

\newaliascnt{remark}{theorem}

\aliascntresetthe{remark}

\newaliascnt{remarks}{theorem}

\aliascntresetthe{remarks}

\newaliascnt{aside}{theorem}

\aliascntresetthe{aside}

\newaliascnt{note}{theorem}

\aliascntresetthe{note}

\usepackage{thmtools}
\usepackage{thm-restate}

\usepackage{hyperref}

\hypersetup{
    colorlinks=true,
    linkcolor=OrangeRed,
    filecolor=Thistle,
    urlcolor=Thistle,
    citecolor=Thistle,
}

\usepackage[nameinlink]{cleveref}

\crefname{theorem}{theorem}{theorems}
\Crefname{theorem}{Theorem}{Theorems}
\crefname{proposition}{proposition}{propositions}
\Crefname{proposition}{Proposition}{Propositions}
\crefname{lemma}{lemma}{lemmas}
\Crefname{lemma}{Lemma}{Lemmas}
\crefname{corollary}{corollary}{corollaries}
\Crefname{corollary}{Corollary}{Corollaries}
\crefname{claim}{claim}{claims}
\Crefname{claim}{Claim}{Claims}
\crefname{definition}{definition}{definitions}
\Crefname{definition}{Definition}{Definitions}
\crefname{example}{example}{examples}
\Crefname{example}{Example}{Examples}
\crefname{assumption}{assumption}{assumptions}
\Crefname{assumption}{Assumption}{Assumptions}

\makeatletter
\let\cref@old@isrefconsecutive\cref@isrefconsecutive
\def\cref@isrefconsecutive#1#2{%
  \begingroup
    \def\cref@assumptiontype{assumption}%
    \cref@gettype{#1}{\cref@typea}%
    \ifx\cref@typea\cref@assumptiontype
      \endgroup
      \@cref@refconsecutivefalse
    \else
      \endgroup
      \cref@old@isrefconsecutive{#1}{#2}%
    \fi
}
\makeatother

\crefname{condition}{condition}{conditions}
\Crefname{condition}{Condition}{Conditions}
\crefname{question}{question}{questions}
\Crefname{question}{Question}{Questions}
\crefname{remark}{remark}{remarks}
\Crefname{remark}{Remark}{Remarks}
\crefname{remarks}{remarks}{remarks}
\Crefname{remarks}{Remarks}{Remarks}
\crefname{aside}{aside}{asides}
\Crefname{aside}{Aside}{Asides}
\crefname{note}{note}{notes}
\Crefname{note}{Note}{Notes}
\crefname{appendix}{appendix}{appendices}
\Crefname{appendix}{Appendix}{Appendices}

\newcommand{\secref}[1]{\hyperref[#1]{\S\ref*{#1}}}

\definecolor{backcolour}{rgb}{0.63,0.79,0.95}
\lstdefinestyle{mystyle}{
  backgroundcolor=\color{backcolour},
  basicstyle=\ttfamily\footnotesize,
  breakatwhitespace=false,
  breaklines=true,
  captionpos=b,
  keepspaces=true,
  numbers=left,
  numbersep=5pt,
  showspaces=false,
  showstringspaces=false,
  showtabs=false,
  tabsize=2
}
\begin{document}
\title{Labels}
\author{Mark Whitmeyer\thanks{Arizona State University. Email: \href{mailto:mark.whitmeyer@gmail.com}{mark.whitmeyer@gmail.com}. For KS. I used ChatGPT as one would an RA and \href{https://www.refine.ink}{refine.ink} for feedback.}}
\date{\today}
\maketitle

\begin{abstract}
Labels--grades, credentials, scores, ratings, ranks--do two things. They inform receivers, and they give agents something to chase. I study optimal classification when labels must be earned through costly self-selection. I show that exact certification is inefficiently fine: pooling a small bottom interval saves first-order signaling costs while losing only higher-order decision value. I provide sufficient conditions for lower censorship to maximize efficiency as well as for every optimal classification to use finitely many categories.
\end{abstract}

\section{Introduction}\label{sec:intro}

Societies make distinctions. Students receive grades. Workers acquire credentials. Borrowers receive credit scores. Restaurants receive stars. Platforms assign badges, follower counts, and verification marks. Journals, schools, firms, and people are ranked. These accolades are informative, but they are also prizes. Once a distinction matters, agents expend resources to be on the better side of it.

I study the design of such labels. The institution first commits to a deterministic classification rule: which differences will be publicly certified, disclosed, or rewarded. Agents then choose costly actions to obtain the induced labels. A receiver observes only the public label.

The usual vocabulary for these institutions is informational: a grade tells us how much a student learned, a credit score tells us how likely a borrower is to repay, a star rating tells us something about quality. That language misses half of the object--a public label also gives agents something to \textit{win}. It changes what the receiver believes, but it also changes what the sender wants to prove. An institution that adds one more certified category is not only adding one more message to the receiver's information set but one more rung to the ladder that agents climb.

This paper asks which rungs are worth building. A sender has a private type, which is its quantile \(q \in [0,1]\), and type \(q\) has quality \(x(q)\), with \(x\) increasing. The receiver has a reduced-form payoff of \(r\) in her posterior mean of the sender's quality, for which I only assume smoothness. The sender's preferences are monotone: higher posterior means are more valuable, so if the institution certifies that one group is better than another, it must make the better certificate costly enough that lower types do not imitate it. The design problem is to decide which distinctions are worth sustaining, not how much information to reveal in the abstract.

A designer chooses a deterministic certification technology, which maps actions into public certificates. I argue that in equilibrium, a certificate is payoff-relevant through two numbers: the posterior mean \(M(q)\) attached to the certificate obtained by type \(q\), and the least action \(z(q)\) needed to obtain that certificate. Bayes-plausibility (that beliefs are a martingale) plus incentive compatibility allow me to reduce the problem of choosing an optimal classification to a choice of a monotone partition, which I term an \textit{admissible} signal.

With this in hand, I next derive a simple accounting formula. For any admissible signal \(M\), I reveal in \Cref{prop:general-cost} that the least implementation cost of \(M\) can be naturally written in terms of the ``price'' of a category boundary. This is the mimicry pressure at the boundary multiplied by the private prize jump across it, and this price is the basic object in the paper. A boundary may be very useful to the receiver, or almost useless; it may create a large status prize, or a small one. The formula I derive in the proposition puts those two forces in the same units.

The formula also immediately points to a simple improvement over full certification. Start from the system that certifies every type exactly. Now pool a small interval at the bottom and fully certify everyone above it. The receiver loses little from this change: the types being pooled are close together, so the loss from averaging their qualities is higher order. The signaling saving is much larger. Eliminating the first certified distinction relaxes the action requirements faced by a positive mass of higher types. In \Cref{thm:generic-coarseness}, I make this comparison precise: under smoothness and positive entry stakes, bottom pooling saves first-order signaling cost while losing only higher-order receiver value. Reassuringly, I show in \Cref{prop:genval} that such bottom pooling improves upon full transparency provided the receiver's payoff is continuous in her posterior belief.

This paper is not meant to rediscover the broad fact that coarse or noisy communication can be valuable. That point has already been made in many forms. Coarse certification and grading can change who seeks certification, how much quality is produced, or how much information is acquired;\footnote{See, e.g., \citet{Lizzeri1999}, \citet{AlbanoLizzeri2001}, \citet{BoleslavskyCotton2015}, \citet{Zubrickas2015}, \citet{HarbaughRasmusen2018}, \citet{OstrovskySchwarz2010}, \citet{DeMarzoKremerSkrzypacz2019}, \citet{BizzottoVigier2021}, \citet{Zapechelnyuk2020}, and \citet{AsseyerWeksler2024}.} signals are themselves be sold, priced, or packaged by intermediaries;\footnote{As in \citet{Rayo2013}, \citet{Lu2025}, and \citet{CamboniNiuPaiVohra2025}. Also related is signaling with commitment \citep{BoleslavskyShadmehr2025}, in which a sender commits \textit{ex ante} to her costly signaling.} noisy or mediated messages can relax communication constraints;\footnote{As studied by \citet{BlumeBoardKawamura2007}, \citet{GoltsmanHornerPavlovSquintani2009}, \citet{Rick2013}, \citet{Whitmeyer2019}, \citet{PerezRichetSkreta2022}, and \citet{Salamanca2021}.} scores, tests, and data-use rules can be designed to blunt manipulation, falsification, or distorted incentives;\footnote{As in \citet{Ball2025}, \citet{FrankelKartik2019,FrankelKartik2022}, and \citet{BoleslavskyKim2021}.} and rank or status disclosures can change contest effort.\footnote{\citet{MoldovanuSelaShi2007}, \citet{OlszewskiSiegel2016}, \citet{OlszewskiSiegel2020}, \citet{KrishnaLychaginOlszewskiSiegelTergiman2026}, and \citet{Goel2025} are notable examples of this line of study.} Nor is my aim to show that signaling is wasteful.\footnote{Nor to quantify this waste, as in \citet{FrankelKartik2026}. My question is complementary: which public distinctions should survive when the institution can choose the classification rule?} Instead, I show that a robust improvement upon full certification is to erase, or blur, the bottom rung.

This result also lends \Cref{thm:generic-coarseness} a costly-certification interpretation. After pooling \([0,\varepsilon]\), the institution can implement the outcome as a free default label plus a certified track. Types below \(\varepsilon\) take the default. To receive any label above the default, a sender must first take an entry action that makes the cutoff type indifferent between the default label and the first certified label.\footnote{This resembles disclosure models, in which disclosure, testing, or  certification is (exogenously) costly.\citep{Verrecchia1983,DeMarzoKremerSkrzypacz2019,Zapechelnyuk2020}.} In this sense I endogenize the \citet{Verrecchia1983}-like threshold: the default label and the entry action are consequences of optimal classification.

The next question is whether this bottom-rung repair is actually optimal (and not just an improvement on full transparency). \Cref{prop:lower-censorship} provides conditions under which it, indeed, constitutes the whole design. A sufficient condition is that mimicry pressure \(L\) is convex in quantile, receiver value \(r\) is convex in posterior mean, and sender stakes \(b\) are weakly concave in certified quality. Since \(L(1)=0\), convexity of \(L\) makes mimicry pressure fall in a disciplined way, and the sender-side term works in the same direction as the receiver-side convexity. Every noninitial pooled interval can then be opened up without lowering welfare. As a result, the only pooling that may be optimal is the initial one: a free default category followed by full certification above a cutoff, \textit{viz.}, lower censorship.\footnote{\Cref{prop:postanalog} reveals that this finding persists for general (continuous-in-belief) receiver preferences.}

I finish the paper with two additional exercises. First, when does the designer optimally choose only finitely many categories? In \Cref{thm:finite-partitions}, I show that, if sufficiently small binary distinctions are locally wasteful, every optimum has finite (essential) range. Finally, in \Cref{prop:contest}, I reveal that the same basic forces manifest when the sender's private prize is generated by a winner-take-all tournament. There, the identical cost-of-distinctions formula applies, and the familiar bottom-pooling improvement over full certification follows.

\smallskip

\noindent \textbf{Roadmap.} \secref{sec:model} introduces the environment and contains preliminary results. \secref{sec:cost} proves the cost-of-distinctions formula, and \secref{sec:coarse} reveals the superiority of pooling at the bottom. \secref{sec:posterior} extends the results to more general receiver preferences, and \secref{sec:finite} asks when finitely-many certifications is optimal. \secref{sec:tournaments} highlights that the same insights persist in signaling tournaments. Omitted proofs inhabit \Cref{app:omit}.

\section{The Signaling Environment}\label{sec:model}

There are three parties: a sender, a receiver, and a third-party information designer. The sender has a privately observed \textit{type}, indexed by its quantile \(q \in \left[0,1\right]\).\footnote{Accordingly, the population distribution is uniform in quantile space: \(Q\sim U[0,1]\). Equivalently, we could start with a type \(\theta\in[\ell,u]\) distributed according to a strictly increasing cdf \(F\), and set \(q=F(\theta)\). In that case \(Q=F(\theta)\) is uniform on \([0,1]\), and \(x(q)\) is the quality of the \(q\)-quantile type. If quality equals type, then \(x(q)=F^{-1}(q)\).} Type \(q\) has \textit{quality} \(x(q)\), where \(x \colon \left[0,1\right] \to\mathbb R\) is continuous and strictly increasing. The sender chooses an \textit{action} \(a \in \left[0,\infty\right)\). Actions are costly: type \(q\) incurs a cost \(a S(x(q))\) from choosing \(a\), where \(S(\cdot) > 0\) is strictly decreasing.

The receiver is passive: given a posterior mean \(m\) of quality, the receiver obtains value \(r(m)\), and the sender obtains private benefit \(b(m)\). The sender's payoff is additively separable in benefit and cost: a type \(q\) sender who induces posterior mean \(m\) via action \(a\) receives \(b(m) - aS(x(q))\).

For the main scalar results I impose the following regularity.
\begin{assumption}\label{ass:scalar}
The quality function \(x\) is \(C^2\) and strictly increasing on
\([0,1]\), with \(x'(q)>0\) for all \(q\). The functions \(b\), \(r\), and \(S\) are \(C^2\) on neighborhoods of \(x([0,1])\), with \(b'>0\), \(S>0\), and \(S'<0\).
\end{assumption}
The smoothness is stronger than is needed for some results, but it makes things easier.

Before the sender chooses an action, the designer commits to a (deterministic) public certification technology, which maps actions into public certificates. The receiver observes the public certificate, not the action itself. After observing the technology, the sender chooses an action, the public certificate is generated, the receiver forms a posterior mean, and payoffs are realized.

Formally, let \(Y\) denote a set of public certificates. A \textit{certification technology} is a function \(\gamma\colon [0,\infty)\to Y\). Fixing \(\gamma\), an equilibrium consists of an action rule \(a\colon[0,1]\to[0,\infty)\) and posterior means for the certificates reached under that action rule. If \(y(q)=\gamma(a(q))\), let \(m(y)\) denote the conditional expectation \(\E[x(Q)\mid y(Q)=y]\). I write \(M(q)=m(y(q))\) for the posterior mean of the certificate obtained by type \(q\).

Since actions are totally ordered by cost and the receiver observes only the certificate, no type pays more than the least action needed to obtain its certificate. Thus, each reached certificate is payoff-relevant through two numbers: the posterior mean it produces and the least action needed to obtain it. The names of certificates are irrelevant, so I work with the induced type-space representation. An equilibrium induces a pair \((M,z)\), where \(M(q)\) is the posterior mean obtained by type \(q\), and \(z(q)\) is the least action needed to obtain type \(q\)'s certificate. Type \(q\)'s equilibrium payoff is, consequently, \(b(M(q))-z(q)S(x(q))\).

At equilibrium, the pair \((M,z)\) must satisfy Bayes plausibility (BP)
\[
        M=\E[x\mid M]\qquad\text{a.s.}
        \tag{BP}\label{eq:BP}
\]
and incentive compatibility (IC)
\[
        b(M(q))-z(q)S(x(q))
        \ge
        b(M(t))-z(t)S(x(q))
        \qquad \text{for a.e. }q,t\in[0,1].
        \tag{IC}\label{eq:IC}
\]

\subsection{Preliminaries}

As a first step, I simplify the designer's problem. In particular, \eqref{eq:IC} plus strict monotonicity of \(b\) and \(S\) imply that higher types obtain (weakly) higher posterior means.

\begin{lemma}\label[lemma]{lem:monom}
In any equilibrium, \(M\) is weakly increasing.
\end{lemma}

As a result, I describe certification outcomes by their induced posterior-mean signals. The relevant \textit{admissible} signals are those that are nondecreasing and Bayes-plausible:
\[
        \mathcal M
        =
        \{M\colon[0,1]\to[x(0),x(1)]:
        M\text{ is nondecreasing and }M=\E[x\mid M]\text{ a.s.}\}.
\]
Full certification is \(M(q)=x(q)\). Full pooling is \(M(q)\equiv\int_0^1x(t)dt\). A \textit{finite classification} is an admissible signal \(M\in\mathcal M\) with finite essential range.\footnote{\textit{Viz.}, there exist finitely many numbers \(m_1,\ldots,m_N\) such that \(M(q)\in\{m_1,\ldots,m_N\}\) for a.e. \(q\).} If \(M\) is finite-valued, monotonicity implies that its cells are intervals up to null sets, and Bayes-plausibility implies that each value of \(M\) is the mean quality in the corresponding cell. For \(M\in\mathcal M\), set \(B(q)=b(M(q))\).\footnote{\label{fn:rep} I identify functions of \(q\) that agree a.e. When such a function has a monotone version, I use its right-continuous monotone version. In particular, because \(M\) is monotone and \(b\) is increasing, I represent \(B=b\circ M\) this way. Thus, when I write \(dB\), I mean the Stieltjes measure generated by this representative: for \(0\le a<b\le1\), \(dB((a,b])=B(b)-B(a)\). In particular, a jump of \(B\)
at \(q>0\) contributes mass \(B(q)-B(q-)\) at \(q\). The same convention applies below when \eqref{eq:IC} implies that an action requirement \(z\) is monotone a.e.--changing \(z\) on a null set does not change the expected implementation cost.} Since \(M\) is nondecreasing and \(b\) is increasing, \(B\) is nondecreasing.

\begin{proposition}\label[proposition]{prop:technology-representation}
Every equilibrium of a deterministic certification technology induces a pair \((M,z)\) such that \(M\in\mathcal M\) and \((M,z)\) satisfies \eqref{eq:IC}. Conversely, for any \(M\in\mathcal M\), there exists a deterministic certification technology \(\gamma\colon\left[0,\infty\right)\to M(\left[0,1\right])\), such that \(a(q)=z^*(q)\coloneqq\int_{\left(0,q\right]}\frac{dB(u)}{S(x(u))}\) is an equilibrium, the induced posterior-mean signal is \(M\) a.s., and the least action needed to obtain type \(q\)'s certificate is \(z^*(q)\) a.s.
\end{proposition}

Thus, moving from deterministic certification technologies to posterior-mean signals is without loss of generality for the design problem. The designer's problem can be written first over equilibrium pairs \((M,z)\). Before minimizing over action costs, the objective is
\[
        \sup_{(M,z)}
        \left\{
        \int_0^1 r(M(q))dq
        -
        \int_0^1 z(q)S(x(q))dq
        \right\},
\]
subject to \(M\in\mathcal M\), \(z\ge0\), and \eqref{eq:IC}. The next step is to solve out the least action cost needed to implement a given posterior-mean signal \(M\). Thanks to \Cref{prop:technology-representation}, once this least cost is solved out, every admissible \(M\) is induced by a deterministic certification technology.

\section{The Cost of Distinctions}
\label{sec:cost}

Define the mimicry pressure at quantile \(q\):\[L(q)\coloneqq \frac{\int_q^1S(x(t))dt}{S(x(q))}.\]
An increase in the public prize above \(q\) must be deterred for lower types and paid by higher types. Let \(K(M)\) be the least expected action cost \(\int_0^1 z(q)S(x(q))dq\) over nonnegative measurable \(z\)s satisfying \eqref{eq:IC} for the specified \(M\). Then,
\begin{proposition}
\label[proposition]{prop:general-cost}
The least expected action cost relates to the mimicry pressure via \(K(M) = \int_{(0,1]}L(q)d[b(M(q))]\).
\end{proposition}

I term the least action needed to obtain the certificate received by type \(z(q)\) an \textit{action requirement}. After substituting the minimum implementation cost, the designer's problem becomes
\[\sup_{M\in\mathcal M} \int_0^1r(M(q))dq - \int_{(0,1]}L(q) d[b(M(q))],\]
which nests both full certification and finite classifications, which I briefly discuss shortly after proving the existence of an optimum:

\begin{proposition}\label[proposition]{prop:scalar-existence}
The designer's problem has a solution.
\end{proposition}

\medskip

\noindent \textbf{Full Certification.} If \(M(q)=x(q)\), then
\[
        K(M)
        =
        \int_{(0,1]}L(q) d[b(x(q))],
\]
or, when \(b\) and \(x\) are differentiable,
\[
        K(M)
        =
        \int_0^1L(q)b'(x(q))x'(q)dq.
\]

\medskip

\noindent \textbf{Finite Classifications.} If \(M\) is finite-valued with boundary quantiles \(0=q_0<q_1<\cdots<q_N=1\), cell means
\[m_j=\frac{1}{q_j-q_{j-1}}
        \int_{q_{j-1}}^{q_j}x(s)ds,
        \qquad \text{for } j=1,\ldots,N,
\]
and writing \(B_j=b(m_j)\), \(B\) jumps at \(q_j\) by
\(B_{j+1}-B_j\). Hence, the minimum-cost formula becomes
\[
        K(M)
        =
        \sum_{j=1}^{N-1}
        L(q_j)\left[b(m_{j+1})-b(m_j)\right].
\]

The cost of a public distinction is the sender-benefit jump across the distinction multiplied by the mimicry pressure at the boundary.

\section{An Easy Improvement}\label{sec:coarse}

As we now discover, smoothness renders full certification suboptimal. The reason is an order comparison: smooth receiver value gives tiny distinctions only higher-order value, while the first layer of the certification ladder has first-order cost.

Let \(M^F(q)=x(q)\) denote full certification. For \(\varepsilon>0\), let \(M^\varepsilon\) pool \(\left[0,\varepsilon\right]\) and fully certify all types above \(\varepsilon\). Write
\[
        m_\varepsilon\coloneqq \frac{1}{\varepsilon}\int_0^\varepsilon x(q)dq.
\]
Thus, \(M^\varepsilon(q)=m_\varepsilon\) on \(\left[0,\varepsilon\right]\), and \(M^\varepsilon(q)=x(q)\) on \(\left(\varepsilon,1\right]\). Let \(R^F\) and \(K^F\) denote the receiver value and implementation cost under full certification, and let \(R^\varepsilon\) and \(K^\varepsilon\) denote their counterparts under \(M^\varepsilon\).

\begin{theorem}\label{thm:generic-coarseness}
If \(L(0)b'(x(0))x'(0)>0\),\footnote{Observe that this follows from \Cref{ass:scalar}.} then there exists \(\bar\varepsilon>0\) such that, for every \(\varepsilon\in\left(0,\bar\varepsilon\right)\),
\(R^\varepsilon-K^\varepsilon>R^F-K^F\), i.e., full certification is not welfare optimal.
\end{theorem}

There is a natural costly-certification interpretation for this bottom pooling. The designer offers a free default label. Types who do not enter the track receives a default label, and the receiver assigns it posterior mean \(m_\varepsilon\). To obtain any certificate strictly above the default label, a sender must first choose
the entry action
\[
\xi_\varepsilon
:=
\frac{b(x(\varepsilon))-b(m_\varepsilon)}{S(x(\varepsilon))}.
\]
The least-cost action schedule implementing \(M^\varepsilon\) is, up to the
irrelevant value at the cutoff,
\[
z^\varepsilon(q)=0 \quad \text{for } q<\varepsilon, \quad \text{and} \quad
z^\varepsilon(q)=
\xi_\varepsilon+
\int_\varepsilon^q
\frac{b'(x(u))x'(u)}{S(x(u))}\,du
\quad \text{for } q>\varepsilon .
\]
As a result, the cutoff type is just indifferent between the default label and the first-certified label, with types below \(\varepsilon\) facing a higher effective cost of the same entry action and remaining with the default label and types above \(\varepsilon\) entering the certified track before then being separated by the rest of the action schedule.

This also highlights why the bottom is special--why not pool at the top, for instance? Eliminating the first certified distinctions
relaxes action requirements for a positive mass of higher types. Pooling the top interval \([1-\varepsilon,1]\) generally does not have the same force. Since \(L(1)=0\), mimicry pressure near the top vanishes. Top pooling, therefore, saves only higher-order cost, and so its welfare effect is curvature-dependent.

\subsection{Lower censorship}

In fact, there is a natural sufficient condition under which it is optimal to simply pool a single interval at the bottom, then fully certify the remainder. We say that \textit{lower-censorship is optimal} if there exists \(c^* \in [0,1]\) such that pooling \(\left[0,c^\ast\right]\) and fully certifying all quantiles above \(c^\ast\) is optimal.

For \(c>0\), define
\[
H(c)\coloneqq cr(m(0,c))+L(c)b(m(0,c)),
\]
and set \(H(0)\coloneqq L(0)b(x(0))\).

\begin{proposition}\label[proposition]{prop:lower-censorship}
Posit \Cref{ass:scalar}. If, for every \(q\in\left(0,1\right)\) and every \(y\in x\left(\left[0,1\right]\right)\),
\(L''(q)\ge0\) and \(r''(y)+L'(q)b''(y)\ge0\), then lower-censorship is optimal.
\end{proposition}

A stronger sufficient condition is economically natural:
\begin{corollary}\label[corollary]{cor:convex}
    If \(L''(q) \geq 0\) and \(r''(y) \geq 0 \geq b''(y)\), then lower-censorship is optimal.
\end{corollary}

That is, lower-censorship is optimal, provided mimicry pressure is convex in quantile, the receiver decision value is convex in posterior mean (which holds if it is an indirect value function for a decision problem), and sender stakes are weakly concave in the certified quality (diminishing returns to its perceived value). Under these conditions, every noninitial pooled interval is inferior to full certification, so the only possible pooling is the bottom one.

\section{Beyond Posterior Means}\label{sec:posterior}

\Cref{thm:generic-coarseness}'s finding that lower-censorship (pooling certificates at the bottom) improves efficiency does not require the receiver's payoff to be a function of the posterior mean and easily generalizes to a significant degree. It is enough that receiver value be continuous when a posterior collapses to the bottom type (holding the sender-side assumptions fixed).

Specifically, let \(V\) be a bounded Borel payoff on posterior distributions over \(x\left(\left[0,1\right]\right)\). Under full certification, type \(q\) induces the posterior \(\delta_{x(q)}\). Under \(M^\varepsilon\), the bottom cell induces the posterior distribution of \(x(Q)\) conditional on \(Q\in[0,\varepsilon]\); call this posterior \(\mu_\varepsilon\). Its mean is \(m_\varepsilon\). In this subsection, write
\[
R^F\coloneqq\int_0^1V(\delta_{x(q)})dq \quad \text{and} \quad
R^\varepsilon\coloneqq\varepsilon V(\mu_\varepsilon)+\int_\varepsilon^1V(\delta_{x(q)})dq,
\]
and understand continuity of \(V\) to mean weak continuity.

\begin{proposition}\label[proposition]{prop:genval}
Maintain the parts of \Cref{ass:scalar} involving \(x,b,S\). Suppose \(V\) is continuous at \(\delta_{x(0)}\). If \(L(0)b'(x(0))x'(0)>0\), then there exists \(\bar{\varepsilon}>0\) such that, for every \(\varepsilon\in(0,\bar{\varepsilon})\), \(R^\varepsilon-K^\varepsilon>R^F-K^F\), i.e., full certification is not welfare optimal.
\end{proposition}

A standard receiver decision problem generates such a \(V\):
\begin{corollary}
Suppose that, after observing the certificate, the receiver chooses \(\alpha\in A\), where \(A\) is a compact metric space. Let the receiver have continuous utiliy \(u\colon A\times x\left(\left[0,1\right]\right)\to\mathbb R\). If \(L(0)b'(x(0))x'(0)>0\), then there exists \(\bar{\varepsilon}>0\) such that, for every \(\varepsilon\in(0,\bar{\varepsilon})\), \(R^\varepsilon-K^\varepsilon>R^F-K^F\), i.e., full certification is not welfare optimal.
\end{corollary}

\Cref{prop:lower-censorship,cor:convex} also possess posterior-valued analogs. For \(c>0\), define
\[
H_V(c)\coloneqq cV(\mu_{0,c})+L(c)b(m(0,c)),
\]
and set \(H_V(0)\coloneqq L(0)b(x(0))\). Define the welfare of lower-censorship at cutoff \(c\) by
\[
\Phi_V(c)\coloneqq H_V(c)+\int_c^1\left[V(\delta_{x(q)})+L'(q)b(x(q))\right]dq.
\]

\begin{proposition}\label[proposition]{prop:postanalog}
Maintain the parts of \Cref{ass:scalar} involving \(x,b,S\). Suppose \(V\) is bounded, weakly continuous, and convex on posteriors over \(x([0,1])\). If \(L''(q)\ge0\) and \(b''(y)\le0\) for every \(q\in(0,1)\) and every \(y\in x([0,1])\), then lower-censorship is optimal.
\end{proposition}

\begin{corollary}
Suppose that, after observing the certificate, the receiver chooses \(\alpha\in A\), where \(A\) is a compact metric space. Let the receiver have continuous utility \(u\colon A\times x([0,1])\to\mathbb R\). If \(L''(q)\ge0\) and \(b''(y)\le0\) for every \(q\in(0,1)\) and every \(y\in x([0,1])\), then lower-censorship is optimal. 
\end{corollary}

\section{Local Distinctions and Finite Optimality}\label{sec:finite}

I now provide a sufficient condition under which every optimum has finitely many categories. Specifically, I show that if, on every sufficiently small interval, every binary split is worse than pooling, then any optimal signal is finite-valued. The economic force is the same as in the cost-of-distinctions formula of \Cref{prop:general-cost}. A new boundary creates a finer message for the receiver, but it also creates a new prize for senders. Types just above the boundary obtain a better certificate than types just below it, and that prize must be sustained by costly self-selection. On a sufficiently small interval, the receiver is distinguishing nearby qualities, while the sender side still treats the boundary as a payoff-relevant rank. My condition is, therefore, that at sufficiently fine scales, this local prize is too expensive to be worth creating.

For \(0\le a<c\le1\), define \(m(a,c)\coloneqq\frac{1}{c-a}\int_a^cx(q)dq\). For an interval \(I=\left[a,c\right]\), write \(\E_I\) for the expectation under the uniform distribution on \(I\). An admissible refinement of \(I\) is a nondecreasing posterior-mean signal on \(I\) satisfying \(M=\E_I[x\mid M]\) a.s. Pooling \(I\) is the constant signal \(M(q)=m(a,c)\).

Define
\[
        G(a,c)\coloneqq(c-a)r(m(a,c))+\left[L(c)-L(a)\right]b(m(a,c)).
\]
The first cell, \(\left[0,c\right]\), has an extra term in the welfare accounting, so recall
\[
        H(c)\coloneqq cr(m(0,c))+L(c)b(m(0,c)).
\]
Recall also the designer's objective
\[
        W(M)\coloneqq\int_0^1r(M(q))dq-\int_{\left(0,1\right]}L(q)d[b(M(q))],
\]
and its finite-classification form: if \(0=q_0<q_1<\cdots<q_N=1\), then
\[
        W(M)=H(q_1)+\sum_{j=2}^NG(q_{j-1},q_j).
\]

Accordingly, \(G(a,c)\) is the net payoff from treating a noninitial interval \(\left[a,c\right]\) as one category. Splitting it at \(d\) creates one additional boundary and, thus, one additional local prize. The comparison \(G(a,d)+G(d,c)-G(a,c)\) is the net value of that prize: receiver precision on the two sides minus the extra implementation cost required to sustain the distinction. The function \(H\) is the analogous payoff for the initial interval, where the bottom of the certification ladder includes an endpoint term.

A binary split is the smallest possible act of categorization: one pooled interval is replaced by two adjacent pooled intervals. The next assumption says that every sufficiently fine version of this act destroys surplus.
\begin{assumption}\label{ass:5}
There exists \(\delta\in\left(0,1\right]\) such that \(G(a,d)+G(d,c)<G(a,c)\) for every \(0<a<d<c\le1\) with \(c-a<\delta\), and \(H(d)+G(d,c)<H(c)\) for every \(0<d<c<\delta\).
\end{assumption}

Using integration by parts, changing \(M\) only on an interval \(I=\left[a,c\right]\) with \(a>0\) changes welfare only through
\[
        \mathcal W_I(M)
        \coloneqq
        \int_a^cr(M(q))dq+\int_a^cb(M(q))L'(q)dq.
\]
Pooling \(I\) delivers \(G(a,c)\). For a bottom interval \(\left[0,c\right]\), define
\[
        \mathcal W^0_{\left[0,c\right]}(M)
        \coloneqq
        \int_0^cr(M(q))dq
        +
        L(0)b(\underline M)
        +
        \int_0^cb(M(q))L'(q)dq,
\]
where \(\underline M\coloneqq\operatorname*{ess\,inf}_{q\in\left[0,c\right]}M(q)\). Pooling \(\left[0,c\right]\) delivers \(H(c)\).

A first step is to move from one extra boundary to many. A finite refinement of a short interval creates a ladder of local prizes. If the first new rung is already wasteful, and the remaining rungs refine a shorter tail of the same interval, then the entire ladder is wasteful. This induction constitutes the next lemma.
\begin{lemma}\label[lemma]{lem:finite-refinement-pooling}
Posit \Cref{ass:5}. If \(0<a<c\le1\) and \(c-a<\delta\), then for every finite refinement \(a=t_0<t_1<\cdots<t_N=c\),
\(\sum_{j=1}^NG(t_{j-1},t_j)\le G(a,c)\), with strict inequality whenever \(N\ge2\). Similarly, if \(0<h<\delta\), then for every finite refinement \(0=t_0<t_1<\cdots<t_N=h\), \(H(t_1)+\sum_{j=2}^NG(t_{j-1},t_j)\le H(h)\), with strict inequality whenever \(N\ge2\).
\end{lemma}
The next step ports the analysis to the continuum--it allows the signal inside a short interval to have infinitely many posterior means. Economically, this is a continuum of tiny status differences rather than a finite ladder. The argument is to group nearby posterior means, apply the finite result to the resulting finite ladder, and then let the groups shrink. This reveals that a continuum of tiny distinctions inherits the same local cost problem as a finite collection of tiny distinctions.

\begin{lemma}\label[lemma]{lem:small-block-pooling}
Posit \Cref{ass:5}. Let \(I=\left[a,c\right]\), with \(a>0\) and \(c-a<\delta\). Then every admissible refinement \(M\) of \(I\) satisfies \(\mathcal W_I(M)\le G(a,c)\), with strict inequality whenever \(M\) is nonconstant. Likewise, if \(0<h<\delta\), then every admissible refinement \(M\) of \(\left[0,h\right]\) satisfies \(\mathcal W^0_{\left[0,h\right]}(M)\le H(h)\), with strict inequality whenever \(M\) is nonconstant.
\end{lemma}

The final step is a counting argument. An infinite classification is an infinitely long prize ladder. In a unit mass of types, sufficiently many rungs force two neighboring rungs to lie in a short quantile interval. Their union is then a short block containing a genuine distinction. But from \Cref{lem:small-block-pooling}, that distinction can be removed profitably.
\begin{theorem}\label{thm:finite-partitions}
Under \Cref{ass:5}, every optimal admissible signal has finite essential range, i.e., is a.e. equivalent to a finite monotone classification. Moreover, the number of categories is at most \(1+2/\delta\).
\end{theorem}

\section{Signaling tournaments}\label{sec:tournaments}

The same bottom-pooling improvement we identified in \secref{sec:coarse} applies when the sender prize is generated by a tournament. Let us now replace \(b(M)\) with a (winner-take-all) tournament prize. Fix \(n\ge2\) and \(s>0\). There are \(n\) candidates with independent quantiles \(q_i\sim U\left[0,1\right]\). The designer commits to a common monotone classification. Each candidate faces the same action technology as above: type \(q_i\) pays \(aS(x(q_i))\) for action \(a\). The candidate with the highest public label receives prize \(s\), and if several candidates share the highest label, they split the prize equally. The receiver value remains function \(r\) of the posterior mean attached to the selected label.

Recall that \(M^F(q)=x(q)\) denotes full certification, and that for  \(\varepsilon>0\), \(M^\varepsilon\) pools \(\left[0,\varepsilon\right]\) and fully certifies all quantiles above \(\varepsilon\). Let \(R^F\) and \(K^F\) now denote receiver value and total implementation cost across all \(n\) candidates under \(M^F\), and let \(R^\varepsilon\) and \(K^\varepsilon\) denote their counterparts under \(M^\varepsilon\).

The proof of \Cref{prop:general-cost} applies to any induced nondecreasing private-prize process \(B\), and so the single-candidate implementation cost is
\(\int_{\left(0,1\right]}L(q)dB(q)\), where we represent \(B\) right-continuously. 
\begin{proposition}\label[proposition]{prop:contest}
Suppose \(x\), \(r\), and \(S\) satisfy the regularity conditions in \Cref{ass:scalar}. Then there exists \(\bar{\varepsilon}>0\) such that, for every \(\varepsilon\in\left(0,\bar{\varepsilon}\right)\),
\(R^\varepsilon-K^\varepsilon>R^F-K^F\), i.e., full certification is not welfare optimal in the tournament problem.
\end{proposition}

The social value of distinguishing among the worst types is less than the price of such distinctions. The juice is not worth the squeeze.

\appendix

\section{Omitted Proofs}\label[appendix]{app:omit}

\subsection{Proof of \texorpdfstring{\Cref{lem:monom}}{lemma}}

\begin{proof}[Proof of \Cref{lem:monom}]
Take \(q_L<q_H\) (outside the null set of ordered pairs for which the relevant incentive constraints may fail) and write \(M_L=M(q_L)\), \(M_H=M(q_H)\), \(z_L=z(q_L)\), and \(z_H=z(q_H)\). From \eqref{eq:IC},
\[b(M_H)-b(M_L) \le (z_H-z_L)S(x(q_L)),
 \quad \text{and} \quad b(M_H)-b(M_L)
        \ge
        (z_H-z_L)S(x(q_H)).\]
Because \(S(x(q_L))>S(x(q_H))\), these inequalities imply \(z_H\ge z_L\). Consequently, \(b(M_H)\ge b(M_L)\), and as \(b\) is strictly increasing, \(M_H\ge M_L\). Thus, \(M\) is increasing a.e.--by \Cref{fn:rep}, I use its increasing version.
\end{proof}

\subsection{Proof of \texorpdfstring{\Cref{{prop:technology-representation}}}{proposition}}

\begin{proof}[Proof of \Cref{{prop:technology-representation}}]
First, fix an equilibrium of a deterministic certification technology \(\gamma\), and write \(y(q)=\gamma(a(q))\). For type \(q\)'s certificate, type \(q\) chooses the least action needed to obtain it in equilibrium. Indeed, if type \(q\) obtains \(y(q)\) and \(a(q)\) were strictly above
\[
\inf\left\{\alpha\ge0\colon\gamma(\alpha)=y(q)\right\},
\]
then there would be some \(\alpha<a(q)\) with \(\gamma(\alpha)=y(q)\), and type \(q\) could strictly lower its cost without changing its certificate. Let \(z(q)\) denote this least action. Then type \(q\)'s payoff is \(b(M(q))-z(q)S(x(q))\).

By the definition of \(M\), \(M(Q)=\mathbb E\left[x(Q)\mid y(Q)\right]\) a.s. Since \(M(Q)\) is measurable with respect to \(y(Q)\),
\[
\mathbb E\left[x(Q)\mid M(Q)\right]
=\mathbb E\left[\mathbb E\left[x(Q)\mid y(Q)\right]\mid M(Q)\right]
=\mathbb E\left[M(Q)\mid M(Q)\right]
=M(Q)
\]
a.s. Thus, \(M=\mathbb E[x\mid M]\) a.s. Type \(q\) can obtain type \(t\)'s certificate by choosing \(z(t)\). Hence, equilibrium optimality implies
\[
b(M(q))-z(q)S(x(q))\ge b(M(t))-z(t)S(x(q))
\]
for a.e. \(q,t\in\left[0,1\right]\). Consequently, \((M,z)\) satisfies \eqref{eq:IC}. \Cref{lem:monom} states that \(M\) is nondecreasing, so \(M\in\mathcal M\).

Conversely, fix \(M\in\mathcal M\), write \(s(q)\coloneqq S(x(q))\), and define \(z^*(q) \coloneqq\int_{\left(0,q\right]}\frac{dB(u)}{s(u)}\). If \(q<t\), then
\[
z^*(t)-z^*(q)=\int_{\left(q,t\right]}\frac{dB(u)}{s(u)}.
\]
Since \(s\) is decreasing, \(s(q)\ge s(u)\ge s(t)\) for \(u\in\left(q,t\right]\). Therefore,
\[
\left[z^*(t)-z^*(q)\right]s(q)\ge B(t)-B(q)\ge\left[z^*(t)-z^*(q)\right]s(t),
\]
which are exactly the two incentive constraints comparing types \(q\) and \(t\). Hence, \((M,z^*)\) satisfies \eqref{eq:IC}.

We now construct a deterministic certification technology. First observe that \(z^*(q)=z^*(t)\) if and only if \(M(q)=M(t)\). Indeed, if \(q<t\) and \(z^*(q)=z^*(t)\), then
\(\int_{\left(q,t\right]}\frac{dB(u)}{s(u)}=0\). Since \(1/s\) is bounded away from zero and \(dB\) is nonnegative, \(B(t)=B(q)\). Since \(b\) is strictly increasing, \(M(t)=M(q)\). The reverse implication follows from the monotonicity of \(M\): if \(M(q)=M(t)\), then \(M\), and, thus, \(B\), is constant on \(\left[q,t\right]\), so \(z^*(q)=z^*(t)\).

Set \(Y\coloneqq M(\left[0,1\right])\), and define \(\gamma\colon\left[0,\infty\right)\to Y\) by
\[
\gamma(\alpha)\coloneqq
\begin{cases}
M(q),&\text{if }\alpha=z^*(q)\text{ for some }q\in\left[0,1\right],\\
M(0),&\text{otherwise.}
\end{cases}
\]
The preceding paragraph reveals that this is well-defined.

Consider the action rule \(a(q)=z^*(q)\). Then \(\gamma(a(q))=M(q)\). Set the posterior mean of every certificate \(y\in Y\) equal to \(y\). As \(\gamma(a(Q))=M(Q)\), Bayes plausibility of \(M\) implies
\[
\mathbb E\left[x(Q)\mid\gamma(a(Q))\right]
=\mathbb E\left[x(Q)\mid M(Q)\right]
=M(Q)
=\gamma(a(Q))
\]
a.s.

It remains to check deviations. Fix type \(q\). If \(\alpha\in z^*(\left[0,1\right])\), choose \(t\) such that \(\alpha=z^*(t)\). Then deviating to \(\alpha\) produces payoff
\[
b(M(t))-\alpha s(q)=B(t)-z^*(t)s(q)\le B(q)-z^*(q)s(q),
\]
where the inequality is \eqref{eq:IC}. If \(\alpha\notin z^*(\left[0,1\right])\), then deviating to \(\alpha\) gives the certificate \(M(0)\), and so yields payoff
\[
b(M(0))-\alpha s(q)\le B(0)\le B(q)-z^*(q)s(q),
\]
where the last inequality follows from
\[
z^*(q)s(q)=s(q)\int_{\left(0,q\right]}\frac{dB(u)}{s(u)}\le B(q)-B(0),
\]
since \(s(q)\le s(u)\) for \(u\le q\). Accordingly, no deviation is profitable, and \(a(q)=z^*(q)\) is an equilibrium.

Finally, \(z^*(q)\) is the least action needed to obtain type \(q\)'s certificate. If \(M(q)=M(0)\), then \(z^*(q)=0\). If \(M(q)>M(0)\) and \(\alpha<z^*(q)\), then \(\gamma(\alpha)\neq M(q)\). Indeed, if \(\alpha\notin z^*(\left[0,1\right])\), then \(\gamma(\alpha)=M(0)\). If instead \(\alpha=z^*(t)\) for some \(t\), then \(\gamma(\alpha)=M(t)=M(q)\) would imply \(z^*(t)=z^*(q)\), contradicting \(\alpha<z^*(q)\). Accordingly, the least action needed to obtain type \(q\)'s certificate is \(z^*(q)\).
\end{proof}

\subsection{Proof of \texorpdfstring{\Cref{prop:general-cost}}{proposition}}

\begin{proof}[Proof of \Cref{prop:general-cost}]
Write \(s(q)=S(x(q))\) and \(H(q)=\int_q^1s(t)dt\), so that \(L(q)=H(q)/s(q)\). \Cref{prop:technology-representation} tells us that \(z^*(q) = \int_{(0,q]}\frac{dB(u)}{s(u)}\) satisfies \eqref{eq:IC} for \(M\) and is induced by a deterministic certification technology.

The expected action cost of \(z^*\) is
\[
        \int_0^1z^*(q)s(q)dq
        =
        \int_0^1
        \left[
        \int_{(0,q]}\frac{dB(u)}{s(u)}
        \right]
        s(q)dq.
\]
By Fubini's theorem,
\[
        \int_0^1z^*(q)s(q)dq
        =
        \int_{(0,1]}
        \frac{\int_u^1s(q)dq}{s(u)}
        dB(u)
        =
        \int_{(0,1]}L(u)dB(u).
\]
Consequently, \(K(M) \leq \int_{(0,1]}L(q)dB(q)\).

It remains to show that no implementation is cheaper. Let \(z\) be any action requirement satisfying \eqref{eq:IC} for this \(M\). Since \eqref{eq:IC} is an a.e. condition, we understand the inequalities below as holding for ordered pairs outside the (null) exceptional set. This is enough for our mesh argument: we can choose the grid points from regular pairs before refining to the limit. For all \(q<t\), \[[z(t)-z(q)]s(t) \leq B(t)-B(q) \leq [z(t)-z(q)]s(q).\]
Since \(B\) is nondecreasing and \(s>0\), the right-hand inequality implies \(z(t)\ge z(q)\). Thus, \(z\) is nondecreasing.

Moreover,
\[\tag{\(A1\)}\label{eq:mesh}
        z(t)-z(q)
        \ge
        \frac{B(t)-B(q)}{s(q)}.
\]
Fix \(0\le a<b\le1\), and take a partition
\[
        a=t_0<t_1<\cdots<t_n=b.
\]
Summing \eqref{eq:mesh} over adjacent partition points, we have
\[z(b)-z(a)
        \ge
        \sum_{i=1}^n
        \frac{B(t_i)-B(t_{i-1})}{s(t_{i-1})}.
\]
Letting the mesh of the partition go to zero yields
\[
        z(b)-z(a)
        \ge
        \int_{(a,b]}\frac{dB(u)}{s(u)}.
\]
Equivalently, as Stieltjes measures, \(dz(u) \geq \frac{dB(u)}{s(u)}\).

Since \(z\) is nondecreasing,
\[
        z(q)
        =
        z(0)+\int_{(0,q]}dz(u).
\]
Using Fubini's theorem again,
\[
        \int_0^1z(q)s(q)dq
        =
        z(0)H(0)
        +
        \int_{(0,1]}H(u)dz(u).
\]
Because \(z(0)\ge0\), \(H(u)\ge0\), and \(dz(u)\ge dB(u)/s(u)\),
\[
        \int_0^1z(q)s(q)dq
        \ge
        \int_{(0,1]}H(u)\frac{dB(u)}{s(u)}
        =
        \int_{(0,1]}L(u)dB(u).
\]
Thus, every implementation costs at least \(\int_{(0,1]}L(q)dB(q)\), while \(z^*\) attains this cost. This proves the formula.
\end{proof}

\subsection{Proof of \texorpdfstring{\Cref{prop:scalar-existence}}{proposition}}

\begin{proof}[Proof of \Cref{prop:scalar-existence}]
The set \(\mathcal M\) is nonempty, since the constant signal \(M(q)=\int_0^1x(t)dt\) is admissible.

For \(M\in\mathcal M\), let \(\underline M\coloneqq\operatorname*{ess\,inf}_{q\in\left[0,1\right]}M(q)\). Using the cost-of-distinctions formula and integration by parts, we rewrite welfare as
\[
        W(M)
        =
        \int_0^1r(M(q))dq
        +
        L(0)b(\underline M)
        +
        \int_0^1b(M(q))L'(q)dq.
\]
Since \(b\) and \(r\) are bounded on \(\left[x(0),x(1)\right]\), and \(L'\) is bounded on \(\left[0,1\right]\), \(W\) is bounded above on \(\mathcal M\). Let \(\{M_n\}\subseteq\mathcal M\) be a sequence such that \(W(M_n) \to \sup_{M \in \mathcal{M}} W(M)\).

Each \(M_n\) is nondecreasing and takes values in the compact interval \(\left[x(0),x(1)\right]\). By Helly's selection theorem, passing to a subsequence if necessary, there is a nondecreasing function \(M\colon\left[0,1\right]\to\left[x(0),x(1)\right]\) such that \(M_n(q)\to M(q)\) at every continuity point of \(M\). Since monotone functions have at most countably many discontinuities, \(M_n\to M\) a.e. Since the sequence is uniformly bounded, by dominated convergence, \(M_n\to M\) in \(L^1\).

We next show that \(M\in\mathcal M\). For every continuous function \(\varphi\colon\left[x(0),x(1)\right]\to\mathbb R\), admissibility of \(M_n\) delivers
\[
        \int_0^1\left(x(q)-M_n(q)\right)\varphi(M_n(q))dq=0.
\]
Passing to the limit by dominated convergence,
\[\tag{\(A2\)}\label{eq:2}\int_0^1\left(x(q)-M(q)\right)\varphi(M(q))dq=0.
\]
Define the finite signed measure \(\zeta\) on \(\left[x(0),x(1)\right]\) by
\[
        \zeta(A)\coloneqq
        \int_{\{q:M(q)\in A\}}
        \left(x(q)-M(q)\right)dq.
\]

Using \eqref{eq:2}, \(\int\varphi d\zeta=0\) for every continuous \(\varphi\). Thus, \(\zeta=0\), which means \(\mathbb E[x-M\mid M]=0\) a.s., so \(M=\mathbb E[x\mid M]\) a.s. Thus, \(M\in\mathcal M\).

It remains to verify upper semicontinuity of \(W\). Since \(M_n\to M\) a.e. and \(b,r\) are continuous and bounded on \(\left[x(0),x(1)\right]\),
\[\tag{\(A3\)}\label{eq:3} \int_0^1r(M_n(q))dq\to\int_0^1r(M(q))dq
\]
and
\[\tag{\(A4\)}\label{eq:4}\int_0^1b(M_n(q))L'(q)dq\to\int_0^1b(M(q))L'(q)dq.
\]
For the endpoint term (\(L(0)b(\underline M)\)), write
\[\underline M_n\coloneqq\operatorname*{ess\,inf}_{q\in\left[0,1\right]}M_n(q).
\]

We claim that \(\limsup_n\underline M_n\le\underline M\). Fix \(\varepsilon>0\). By the definition of essential infimum, the set
\[
        A_\varepsilon\coloneqq
        \{q\in\left[0,1\right]:M(q)<\underline M+\varepsilon\}
\]
has positive measure. If, along a subsequence, \(\underline M_n\ge\underline M+2\varepsilon\), then \(M_n(q)\ge\underline M+2\varepsilon\) a.e., and, hence,
\(|M_n(q)-M(q)|\ge\varepsilon\) for a.e. \(q\in A_\varepsilon\). This contradicts \(M_n\to M\) in \(L^1\). Therefore, \(\limsup_n\underline M_n\le\underline M\). Since \(b\) is increasing and continuous and \(L(0)\ge0\),
\[\tag{\(A5\)}\label{eq:5}\limsup_n L(0)b(\underline M_n)\le L(0)b(\underline M).
\]
Combining \eqref{eq:3}, \eqref{eq:4} and \eqref{eq:5},
\(\limsup_nW(M_n)\le W(M)\), so \(M\) attains the supremum over \(\mathcal M\).
\end{proof}

\subsection{Proof of \texorpdfstring{\Cref{thm:generic-coarseness}}{theorem}}
\begin{proof}[Proof of \Cref{thm:generic-coarseness}]
Under \(M^\varepsilon\),
\[K^\varepsilon = L(\varepsilon)\left[b(x(\varepsilon))-b(m_\varepsilon)\right] + \int_\varepsilon^1L(q)b'(x(q))x'(q)dq.\]
Since \(K^F=\int_0^1L(q)b'(x(q))x'(q)dq\),
\[K^F-K^\varepsilon = \int_0^\varepsilon L(q)b'(x(q))x'(q)dq - L(\varepsilon)\left[b(x(\varepsilon))-b(m_\varepsilon)\right].\]

Choose \(\eta>0\) such that \((1+\eta)^3/2<3/4\). Set \(A(q)\coloneqq L(q)b'(x(q))x'(q)\), and let \(A_0=A(0)>0\). The continuity of \(A\), \(L\), \(x'\), and \(b'\) grants that there exists \(\bar\varepsilon_1>0\) such that
\begin{enumerate}[noitemsep]
    \item\label{it:41} \(A(q)\ge \frac{3}{4}A_0\) for every \(q\in\left[0,\bar\varepsilon_1\right]\);
    \item\label{it:42} \(L(q)\le(1+\eta)L(0)\) and \(x'(q)\le(1+\eta)x'(0)\) for every \(q\in\left[0,\bar\varepsilon_1\right]\); and
    \item\label{it:43} \(b'(y)\le (1+\eta)b'(x(0))\) for every \(y\in\left[x(0),x(\bar\varepsilon_1)\right]\).
\end{enumerate}
From \ref{it:41},
\[\int_0^\varepsilon L(q)b'(x(q))x'(q)dq \ge \frac{3}{4}A_0\varepsilon.
\]

Since \(m_\varepsilon\in\left[x(0),x(\varepsilon)\right]\), using \ref{it:43} and the mean value theorem,
\[b(x(\varepsilon))-b(m_\varepsilon) \leq (1+\eta)b'(x(0))\left[x(\varepsilon)-m_\varepsilon\right].\]
Moreover, from \ref{it:42}
\[
        x(\varepsilon)-m_\varepsilon
        =
        \frac{1}{\varepsilon}\int_0^\varepsilon\left[x(\varepsilon)-x(q)\right]dq
        \le
        \frac{1}{\varepsilon}\int_0^\varepsilon(\varepsilon-q)(1+\eta)x'(0)dq
        =
        \frac{1}{2}(1+\eta)x'(0)\varepsilon.
\]
Therefore,
\[
        L(\varepsilon)\left[b(x(\varepsilon))-b(m_\varepsilon)\right]
        \le
        \frac{1}{2}(1+\eta)^3A_0\varepsilon
        <
        \frac{3}{4}A_0\varepsilon.
\]
Setting \(c\coloneqq \left(\frac{3}{4}-\frac{(1+\eta)^3}{2}\right)A_0>0\), we conclude that \(K^F-K^\varepsilon\ge c\varepsilon\) for every \(\varepsilon\in\left(0,\bar\varepsilon_1\right)\).

On the other hand, the receiver's value difference is
\[R^F-R^\varepsilon = \int_0^\varepsilon r(x(q))dq-\varepsilon r(m_\varepsilon).\]
Let \(M_r\coloneqq\sup_{m\in x\left(\left[0,1\right]\right)}\left|r''(m)\right|\), and let \(C_x\) be a Lipschitz constant for \(x\) on \(\left[0,1\right]\). Taylor's theorem delivers, for every \(q\in\left[0,\varepsilon\right]\),
\[\left|r(x(q))-r(m_\varepsilon)-r'(m_\varepsilon)(x(q)-m_\varepsilon)\right|
        \le
        \frac{M_r}{2}(x(q)-m_\varepsilon)^2.
\]
The linear term (\(r'(m_\varepsilon)(x(q)-m_\varepsilon)\)) integrates to zero by the definition of \(m_\varepsilon\). Since \(\left|x(q)-m_\varepsilon\right|\le C_x\varepsilon\) on \(\left[0,\varepsilon\right]\),
\[
        \left|R^F-R^\varepsilon\right|
        \le
        \frac{M_rC_x^2}{2}\varepsilon^3.
\]
Choose \(\bar\varepsilon\in\left(0,\bar\varepsilon_1\right)\) such that \(M_rC_x^2\varepsilon^2/2<c\) for every \(\varepsilon\in\left(0,\bar\varepsilon\right)\). Then
\[
        \left(R^\varepsilon-K^\varepsilon\right)-\left(R^F-K^F\right)
        =
        \left(K^F-K^\varepsilon\right)-\left(R^F-R^\varepsilon\right)
        \ge
        c\varepsilon-\frac{M_rC_x^2}{2}\varepsilon^3>0.\qedhere
\]
\end{proof}

\subsection{Proof of \texorpdfstring{\Cref{prop:lower-censorship}}{proposition}}
\begin{proof}[Proof of \Cref{prop:lower-censorship}]
Fix \(0<a<c\le1\), and write \(m=m(a,c)\). For each \(q\in\left[a,c\right]\), the function
\(y\mapsto r(y)+L'(q)b(y)\) is convex by assumption. Hence,
\[
r(x(q))+L'(q)b(x(q))-r(m)-L'(q)b(m)\ge\left[r'(m)+L'(q)b'(m)\right]\left[x(q)-m\right].
\]
Integrating over \(\left[a,c\right]\) yields
\[\int_a^c\left[r(x(q))+L'(q)b(x(q))\right]dq-G(a,c) \ge b'(m)\int_a^cL'(q)\left[x(q)-m\right]dq,\]
because \(\int_a^c\left[x(q)-m\right]dq=0\). Moreover,
\[
\int_a^cL'(q)\left[x(q)-m\right]dq
=\frac{1}{2\left(c-a\right)}\int_a^c\int_a^c\left[L'(q)-L'(t)\right]\left[x(q)-x(t)\right]dqdt\ge0,
\]
since \(L'\) and \(x\) are increasing. Therefore,
\[
G(a,c)\le\int_a^c\left[r(x(q))+L'(q)b(x(q))\right]dq
\]
for every \(0<a<c\le1\).

Now take any finite classification with cutoffs \(0=q_0<q_1<\cdots<q_N=1\). By the finite-classification accounting,
\[
W(M)=H(q_1)+\sum_{j=2}^NG(q_{j-1},q_j).
\]
Applying the preceding inequality to every noninitial cell,
\[
\begin{split}
W(M)
&\le H(q_1)+\sum_{j=2}^N\int_{q_{j-1}}^{q_j}\left[r(x(q))+L'(q)b(x(q))\right]dq\\
&=H(q_1)+\int_{q_1}^1\left[r(x(q))+L'(q)b(x(q))\right]dq\\
&\le H(c^\ast)+\int_{c^\ast}^1\left[r(x(q))+L'(q)b(x(q))\right]dq.
\end{split}
\]
The last expression is the welfare of the classification that pools \(\left[0,c^\ast\right]\) and fully certifies all higher quantiles. Thus, no finite classification improves on it.

It remains to extend the bound from finite classifications to arbitrary admissible signals. Let \(M\in\mathcal M\). For each \(n\), partition the essential range of \(M\) into finitely many intervals with mesh tending to zero, let \(\sigma_n\) be the finite sigma-algebra generated by their preimages, and set
\[
M_n\coloneqq E[x\mid\sigma_n]=E[M\mid\sigma_n].
\]
The equality follows from \(M=E[x\mid M]\) and \(\sigma_n\subseteq\sigma(M)\). Since \(M\) is nondecreasing, the preimages of the range bins are quantile intervals up to null sets, so each \(M_n\) is a finite monotone classification. Hence, the finite-classification argument delivers
\[
W(M_n)\le H(c^\ast)+\int_{c^\ast}^1\left[r(x(q))+L'(q)b(x(q))\right]dq
\]
for every \(n\). As the range mesh tends to zero, \(M_n\to M\) in \(L^1\). Since \(b\) and \(r\) are Lipschitz on \(x\left(\left[0,1\right]\right)\), and \(L'\) is bounded, the integral terms in welfare converge. Choosing the bottom range bin with mesh tending to zero also makes the endpoint term converge. Accordingly, \(W(M_n)\to W(M)\), and so the same bound holds for \(M\). Thus, the proposed classification is optimal.
\end{proof}

\subsection{Proof of \texorpdfstring{\Cref{prop:genval}}{proposition}}
\begin{proof}[Proof of \Cref{prop:genval}]
The mean of \(\mu_\varepsilon\) is \(m_\varepsilon\), so the sender side, and, thus, implementation costs are unchanged. The cost comparison in the proof of \Cref{thm:generic-coarseness} provides \(c>0\) and \(\bar{\varepsilon}_1>0\) such that \(K^F-K^\varepsilon\ge c\varepsilon\) for every \(\varepsilon\in(0,\bar{\varepsilon}_1)\).

Concordantly, we need only to bound the receiver term. We have
\[
R^F-R^\varepsilon=\int_0^\varepsilon V(\delta_{x(q)})dq-\varepsilon V(\mu_\varepsilon).
\]
Since \(x\) is continuous, \(\mu_\varepsilon\) converges weakly to \(\delta_{x(0)}\). Indeed, for every continuous \(\varphi\colon x\left(\left[0,1\right]\right)\to\mathbb R\),
\[
\int\varphi d\mu_\varepsilon=\frac{1}{\varepsilon}\int_0^\varepsilon\varphi(x(q))dq\to\varphi(x(0))=\int\varphi d\delta_{x(0)}.
\]
Also, if \(\varepsilon_n\downarrow0\) and \(q_n\in[0,\varepsilon_n]\), then \(q_n\to0\), so \(\delta_{x(q_n)}\) converges weakly to \(\delta_{x(0)}\). As \(V\) is continuous at \(\delta_{x(0)}\), there exists \(\bar{\varepsilon}_2>0\) such that, for every \(\varepsilon\in(0,\bar{\varepsilon}_2)\),
\[
\sup_{q\in[0,\varepsilon]}\left|V(\delta_{x(q)})-V(\delta_{x(0)})\right|+\left|V(\mu_\varepsilon)-V(\delta_{x(0)})\right|\le\frac{c}{2}.
\]

Therefore, for every \(\varepsilon\in(0,\bar{\varepsilon}_2)\),
\[
\left|R^F-R^\varepsilon\right|\le\varepsilon\sup_{q\in[0,\varepsilon]}\left|V(\delta_{x(q)})-V(\delta_{x(0)})\right|+\varepsilon\left|V(\mu_\varepsilon)-V(\delta_{x(0)})\right|\le\frac{c}{2}\varepsilon.
\]
Set \(\bar{\varepsilon}\coloneqq\min\{\bar{\varepsilon}_1,\bar{\varepsilon}_2\}\). Then, for every \(\varepsilon\in(0,\bar{\varepsilon})\),
\[
(R^\varepsilon-K^\varepsilon)-(R^F-K^F)=(K^F-K^\varepsilon)-(R^F-R^\varepsilon)\ge K^F-K^\varepsilon-\left|R^F-R^\varepsilon\right|\ge\frac{c}{2}\varepsilon>0.
\]
Thus, full certification is not welfare optimal.
\end{proof}

\subsection{Proof of \texorpdfstring{\Cref{prop:postanalog}}{proposition}}
\begin{proof}[Proof of \Cref{prop:postanalog}]
As \(V\) is bounded and continuous, \(\Phi_V\) is continuous, so it has a maximizer. Fix \(0<a<c\le1\). Easily,
\(c-a)V(\mu_{a,c})\le\int_a^cV(\delta_{x(q)})dq\), due to \(V\)'s convexity. It remains to compare the sender terms. Since \(L(1)=0\), \(L\ge0\), and \(L''\ge0\), we have \(L'\le0\). Concavity of \(b\) delivers
\(b(x(q))-b(m)\le b'(m)[x(q)-m]\). Multiplying by \(L'(q)\le0\) and integrating over \([a,c]\) yields
\[
\int_a^cL'(q)b(x(q))dq-\left[L(c)-L(a)\right]b(m)\ge b'(m)\int_a^cL'(q)[x(q)-m]dq.
\]
Moreover,
\[
\int_a^cL'(q)[x(q)-m]dq=\frac{1}{2(c-a)}\int_a^c\int_a^c\left[L'(q)-L'(t)\right]\left[x(q)-x(t)\right]dqdt\ge0,
\]
because \(L'\) and \(x\) are increasing. Hence,
\[
(c-a)V(\mu_{a,c})+\left[L(c)-L(a)\right]b(m(a,c))\le\int_a^c\left[V(\delta_{x(q)})+L'(q)b(x(q))\right]dq.
\]

Now take any finite classification with cutoffs \(0=q_0<q_1<\cdots<q_N=1\). Its welfare is
\[
W_V(M)=H_V(q_1)+\sum_{j=2}^N\left[(q_j-q_{j-1})V(\mu_{q_{j-1},q_j})+\left[L(q_j)-L(q_{j-1})\right]b(m(q_{j-1},q_j))\right].
\]
Applying the preceding inequality to every noninitial cell,
\[
W_V(M)\le H_V(q_1)+\int_{q_1}^1\left[V(\delta_{x(q)})+L'(q)b(x(q))\right]dq=\Phi_V(q_1)\le\Phi_V(c^*).
\]
The right-hand side is the welfare of the classification that pools \([0,c^*]\) and fully certifies all higher quantiles.

The extension from finite classifications to arbitrary admissible classifications follows by the same approximation used in the proof of \Cref{prop:lower-censorship}. Take a nested sequence of finite range coarsenings of the induced posterior-mean process, with mesh tending to zero, and let \(M_n\) be the corresponding finite classification. The induced posteriors converge weakly a.s. to the original induced posteriors. Indeed, if \(\sigma_n\) is the sigma-algebra generated by the \(n\)th coarsening, then \(\sigma_n\uparrow\sigma(M)\), so L\'evy's upward theorem \citep[Theorem~4.6.8]{Durrett2019} implies that \(\E[\varphi(x(Q))\mid\sigma_n]\to\E[\varphi(x(Q))\mid\sigma(M)]\) a.s. for each \(\varphi\) in a countable convergence-determining class of continuous functions on \(x([0,1])\). Since \(V\) is bounded and continuous, the receiver term converges by dominated convergence. The sender term converges exactly as in the proof of \Cref{prop:lower-censorship}. Accordingly, \(W_V(M_n)\to W_V(M)\), so the finite-classification bound passes to \(M\) and lower-censorship at \(c^*\) is optimal.
\end{proof}

\subsection{Proof of \texorpdfstring{\Cref{lem:finite-refinement-pooling}}{lemma}}
\begin{proof}[Proof of \Cref{lem:finite-refinement-pooling}]
Consider the \(\left[a,c\right]\) case. We proceed via induction on \(N\). The case \(N=1\) is equality, and the case \(N=2\) is exactly the first inequality in \Cref{ass:5}. For \(N>2\),
\[
\begin{split}
        \sum_{j=1}^NG(t_{j-1},t_j)-G(a,c)
        &=
        \left[G(a,t_1)+G(t_1,c)-G(a,c)\right] \\
        &\qquad+
        \left[\sum_{j=2}^NG(t_{j-1},t_j)-G(t_1,c)\right].
\end{split}
\]
The first bracket is strictly negative by \Cref{ass:5}, and the second is weakly negative by the induction hypothesis applied to \(\left[t_1,c\right]\).

The \(\left[0,h\right]\) case follows the same logic.\end{proof}

\subsection{Proof of \texorpdfstring{\Cref{lem:small-block-pooling}}{lemma}}

\begin{proof}[Proof of \Cref{lem:small-block-pooling}]
First consider \(I=\left[a,c\right]\), with \(a>0\). Let \(M\) be an admissible refinement of \(I\). For each \(n\), form a finite-valued signal \(M_n\) by grouping the values of \(M\) into intervals of length at most \(1/n\), and replacing \(M\) on each group by the average quality of the types in that group. Since \(M=\E_I[x\mid M]\), this is the same as replacing \(M\) on each group by the average value of \(M\) on that group. Since \(M\) is nondecreasing, each group is an interval up to null sets. Thus, \(M_n\) is a finite interval refinement of \(I\).

\Cref{lem:finite-refinement-pooling} delivers \(\mathcal W_I(M_n)\le G(a,c)\). Moreover, on each group, both \(M\) and \(M_n\) lie in an interval of length at most \(1/n\), so \(M_n\to M\) in \(L^1\). Since \(b\) and \(r\) are Lipschitz on \(x\left(\left[0,1\right]\right)\) and \(L'\) is bounded, \(\mathcal W_I(M_n)\to\mathcal W_I(M)\). Hence, \(\mathcal W_I(M)\le G(a,c)\).

If \(M\) is nonconstant, choose a value \(y\) such that both \(\left\{q\in I \colon M(q)\le y\right\}\) and \(\left\{q\in I\colon M(q)>y\right\}\) have positive measure. Since \(M\) is nondecreasing, these sets are intervals up to null sets; write them as \(\left[a,d\right]\) and \(\left[d,c\right]\), with \(a<d<c\). Thus,
\[
        \mathcal W_I(M) \leq G(a,d)+G(d,c)<G(a,c),
\]
where the first inequality is the one we proved in the previous paragraph, applied to each piece, and the second inequality is \Cref{ass:5}.

The proof for the bottom interval \([0,h]\) follows the same logic.
\end{proof}

\subsection{Proof of \texorpdfstring{\Cref{thm:finite-partitions}}{theorem}}
\begin{proof}[Proof of \Cref{thm:finite-partitions}]
Let \(M^*\) be optimal, and suppose first that \(M^*\) has infinite essential range. Then, for every \(N\), we can coarsen the essential range of \(M^*\) into \(N\) (nonnull) intervals. Choose \(N>1+2/\delta\). Since \(M^*\) is nondecreasing, the preimages of these intervals are quantile intervals up to null sets--write their cutoffs as \(0=q_0<q_1<\cdots<q_N=1\). If every adjacent pair had combined length at least \(\delta\), then
\[
        (N-1)\delta
        \le
        \sum_{j=1}^{N-1}\left[\left(q_j-q_{j-1}\right)+\left(q_{j+1}-q_j\right)\right]
        \le
        2,
\]
contradicting \(N>1+2/\delta\). Hence, some adjacent pair has union length below \(\delta\).

Pooling this pair \(\left[a,c\right]\), Bayes plausibility produces
\[
        m(a,c)=\frac{1}{c-a}\int_a^cx(q)dq=\frac{1}{c-a}\int_a^cM^*(q)dq.
\]
Thus, \(m(a,c)\) lies between the essential infimum and essential supremum of \(M^*\) on the block. Since the block is an interval in the ordered essential range of \(M^*\), replacing \(M^*\) on the block by \(m(a,c)\) preserves monotonicity. If \(a=0\), the bottom part of \Cref{lem:small-block-pooling} applies; otherwise, the noninitial part applies. Pooling strictly raises welfare, contradicting optimality. Thus, \(M^*\) cannot have infinite essential range.

Now suppose \(M^*\) has finite essential range with \(N\) values. If \(N>1+2/\delta\), the same counting argument finds an adjacent pair whose union has length below \(\delta\), and \Cref{lem:small-block-pooling} again implies that pooling the pair strictly raises welfare. Therefore, \(N\le1+2/\delta\).
\end{proof}

\subsection{Proof of \texorpdfstring{\Cref{prop:contest}}{proposition}}
\begin{proof}[Proof of \Cref{prop:contest}]
Set \(L_0\coloneqq L(0) > 0\) (\(\Cref{ass:scalar}\)). Under full certification, a candidate of quantile \(q\) wins exactly when all \(n-1\) opponents have lower quantiles. Thus, the induced private prize is \(sq^{n-1}\). Under \(M^\varepsilon\), the bottom label receives prize \(s/n\) exactly when all \(n-1\) opponents also have bottom labels, and receives zero otherwise. Thus, the bottom label has private prize \(s\varepsilon^{n-1}/n\). Every separated quantile \(q>\varepsilon\) still has private prize \(sq^{n-1}\). Applying the cost formula to this induced private-prize process, therefore, proffers
\[
\frac{K^F-K^\varepsilon}{n}=s(n-1)\int_0^\varepsilon L(q)q^{n-2}dq-\frac{s(n-1)}{n}L(\varepsilon)\varepsilon^{n-1}.
\]

Choose \(\eta\coloneqq1/\left(2\left(2n-1\right)\right)\). As \(L\) is continuous at zero, there exists \(\varepsilon_1>0\) such that, for every \(q\in\left[0,\varepsilon_1\right]\),
\[
\left(1-\eta\right)L_0\le L(q)\le\left(1+\eta\right)L_0.
\]
Hence, for every \(\varepsilon\in\left(0,\varepsilon_1\right)\),
\[
\begin{split}
\frac{K^F-K^\varepsilon}{n}
&\ge s(n-1)\left(1-\eta\right)L_0\int_0^\varepsilon q^{n-2}dq-\frac{s(n-1)}{n}\left(1+\eta\right)L_0\varepsilon^{n-1}\\
&=sL_0\left[\frac{1}{n}-\eta\frac{2n-1}{n}\right]\varepsilon^{n-1} = \frac{sL_0}{2n}\varepsilon^{n-1},
\end{split}
\]
and so \(K^F-K^\varepsilon\ge\frac{sL_0}{2}\varepsilon^{n-1}\).

It remains to bound the receiver-value loss. Let \(C\ge1\) be a Lipschitz constant for \(r\circ x\) on \(\left[0,1\right]\). If at least one candidate has quantile above \(\varepsilon\), the same quantile is selected almost surely under \(M^F\) and \(M^\varepsilon\): among candidates above \(\varepsilon\), the highest quantile is still fully certified and beats every bottom label. Accordingly, the selected posterior mean can differ only on the event that all \(n\) candidates lie in \(\left[0,\varepsilon\right]\), which has probability \(\varepsilon^n\).

On that event, let \(q^\ast\) denote the full-certification winner. Then \(q^\ast\in\left[0,\varepsilon\right]\). Under bottom pooling, the selected posterior mean is \(m_\varepsilon\). Since \(m_\varepsilon\in\left[x(0),x(\varepsilon)\right]\), there is \(q_\varepsilon\in\left[0,\varepsilon\right]\) such that \(m_\varepsilon=x(q_\varepsilon)\). Consequently,
\[
\left|r(x(q^\ast))-r(m_\varepsilon)\right|=\left|r(x(q^\ast))-r(x(q_\varepsilon))\right|\le C\varepsilon,\] and so
\(\left|R^F-R^\varepsilon\right|\le C\varepsilon^{n+1}\).

Finally, set
\(\bar{\varepsilon}\coloneqq\min\left\{\varepsilon_1,\left(\frac{sL_0}{4C}\right)^{1/2}\right\}\) so that for every \(\varepsilon\in\left(0,\bar{\varepsilon}\right)\),
\[\left(R^\varepsilon-K^\varepsilon\right)-\left(R^F-K^F\right) \ge\left(K^F-K^\varepsilon\right)-\left|R^F-R^\varepsilon\right| \ge\frac{sL_0}{2}\varepsilon^{n-1}-C\varepsilon^{n+1} > 0. \qedhere\]
\end{proof}

\bibliography{sample}

\end{document}